\documentclass[onecolumn]{IEEEtran}

\usepackage[usenames,dvipsnames]{pstricks}
 \usepackage{pst-grad} 
 \usepackage{pst-plot} 
 \usepackage[space]{grffile} 
 \usepackage{etoolbox} 
 \makeatletter 
 \patchcmd\Gread@eps{\@inputcheck#1 }{\@inputcheck"#1"\relax}{}{}
 \makeatother
  \usepackage{color}

\usepackage[cmex10]{amsmath}
\usepackage{xcolor}
\usepackage[multiple]{footmisc}

\usepackage{mdwtab}
\usepackage{eqparbox}
\usepackage{tabularx}
\usepackage{graphicx,amssymb,rangecite,upgreek,dsfont,mathrsfs}
\usepackage{datetime}
\usepackage{algorithm}
\usepackage{algorithmic} 
\usepackage[usenames,dvipsnames]{pstricks}
 \usepackage{epsfig}

 \usepackage{pst-grad} 
 \usepackage{pst-plot} 

\newcommand {\bE} {\mathbb{E}}

\newcommand{\calX}{{\cal X}}

\newcommand{\de}{\stackrel{.}{=}}

\newcommand{\be}{\begin{equation}}
\newcommand{\ee}{\end{equation}}
\newcommand{\beqna}{\begin{eqnarray}}
\newcommand{\eeqna}{\end{eqnarray}}

\usepackage{theorem}
\DeclareFontFamily{U}{mathx}{\hyphenchar\font45}
\DeclareFontShape{U}{mathx}{m}{n}{
      <5> <6> <7> <8> <9> <10>
      <10.95> <12> <14.4> <17.28> <20.74> <24.88>
      mathx10
      }{}
\DeclareSymbolFont{mathx}{U}{mathx}{m}{n}

\DeclareMathSymbol{\bigtimes}{1}{mathx}{"91}
\theorembodyfont{\rmfamily}

\newcommand{\abs}[1]{\left|#1\right|}

\theoremheaderfont{\itshape}
\newtheorem{definition}{Definition}
\newtheorem{theorem}{Theorem}
\newtheorem{proof}{Proof}

\newtheorem{lemma}{Lemma} 
\newtheorem{corollary}{Corollary}

\newtheorem{remark}{Remark}

\newcommand{\pp}[1]{\left[#1\right]}
\newcommand{\ppp}[1]{\left\{#1\right\}}

\allowdisplaybreaks
\usepackage{comment}
\newcommand{\salman}[1]{\textcolor{black}{#1}}

\newcommand{\rev}[1]{\textcolor{black}{#1}}

\makeatletter
\newcommand{\subalign}[1]{%
	\vcenter{%
		\Let@ \restore@math@cr \default@tag
		\baselineskip\fontdimen10 \scriptfont\tw@
		\advance\baselineskip\fontdimen12 \scriptfont\tw@
		\lineskip\thr@@\fontdimen8 \scriptfont\thr@@
		\lineskiplimit\lineskip
		\ialign{\hfil$\m@th\scriptstyle##$&$\m@th\scriptstyle{}##$\crcr
			#1\crcr
		}%
	}
}
\makeatother

\begin{document}

\title{\rev{Centralized vs Decentralized Targeted Brute-Force Attacks: Guessing with Side-Information}}

\author{\IEEEauthorblockN{Salman Salamatian,}
\and
\IEEEauthorblockN{Wasim Huleihel,} 
\and
\IEEEauthorblockN{Ahmad Beirami,} 
\and
\IEEEauthorblockN{Asaf Cohen,}
\and
\IEEEauthorblockN{Muriel M\'edard}
\IEEEoverridecommandlockouts
\IEEEcompsocitemizethanks{
\IEEEcompsocthanksitem
This work was presented in part at 2017 IEEE Symposium on Information Theory~\cite{salamatian2017centralized}.

S. Salamatian, A. Beirami, and M. M\'edard are with the department of Electrical Engineering and Computer Science, MIT, Cambridge, 02139 MA (\{salmansa, beirami, medard\}@mit.edu).
W. Huleihel is with the Department of Electrical Engineering-Systems at Tel-Aviv University, Tel-Aviv 6997801, Israel (e-mail: wasimh@mit.edu). 
A. Cohen is with the department of Electrical Engineering, Ben-Gurion University of the Negev, 8410501 Israel (coasaf@bgu.ac.il)
}}

\parskip 3pt

\maketitle

\begin{abstract}
    According to recent empirical studies, a majority of users have the same, or very similar, passwords across multiple password-secured online services. 
    This practice can have disastrous consequences, as one password being compromised puts all the other accounts at much higher risk.
    Generally, an adversary may use any side-information he/she possesses about the user, be it demographic information, password reuse on a previously compromised account, or any other relevant information to devise a better brute-force strategy (so called targeted attack).
	
	In this work, we consider a distributed brute-force attack scenario in which $m$ adversaries, each observing some side information, attempt breaching a password secured system. We compare two strategies: an uncoordinated attack in which the adversaries query the system based on their own side-information until they find the correct password, and a fully coordinated attack in which the adversaries pool their side-information and query the system together. For passwords $\mathbf{X}$ of length $n$, generated independently and identically from a distribution $P_X$, we establish an asymptotic closed-form expression for the uncoordinated and coordinated strategies when the side-information $\mathbf{Y}_{(m)}$ are generated independently from passing $\mathbf{X}$ through a memoryless channel $P_{Y|X}$, as the length of the password $n$ goes to infinity. We illustrate our results for binary symmetric channels and binary erasure channels, two families of side-information channels which model password reuse. We demonstrate that two coordinated agents perform asymptotically better than any finite number of uncoordinated agents for these channels, meaning that sharing side-information is very valuable in distributed attacks.
\end{abstract}

\section{Introduction}\label{sec:intro}

Brute-force attacks represent a significant portion of cyber-attacks \cite{quaterlyreport}. They target password-secured systems and consist in querying tentative passwords until the correct one is found. This can take place in an online way, where the adversary connects to a host server, sends her password queries, and receives notification of her success or failure after each guess. More often though, these attack take place offline. In this case, the adversary has previously gained access to a collection of hashed passwords through another breach, and queries tentative passwords by comparing them to a hash. In either case, the number of queries---or guesses---is a surrogate for the computational effort the adversary has to accomplish to breach the system. As such, understanding quantities such as the average number of guesses before the correct password is found, are useful in assessing the security risks of a system against brute-force attacks. This can be modeled by the \emph{guesswork}, which measures the number of queries needed before \emph{guessing} correctly a discrete random variable $X$ with probability mass function (pmf) $P_X$. More precisely, let Alice select a secret sequence of length $n$ denoted by $\mathbf{X} = (X_1,\ldots,X_n)$, where $X_i \in \calX$. Assume further that this sequence is selected at random, such that $\{X_i\}_{i=1}^n$ are independent and identically distributed (i.i.d.), with pmf $P_X$.
Then, Bob, who does not see the realization of $\mathbf{X}$ but does know $P_X$, presents to Alice a successive sequence of guesses $\hat{\mathbf{X}}_1,\hat{\mathbf{X}}_2$, and so on.
\salman{For each guess $\hat{\mathbf{X}}_i$, Alice checks whether it is the correct sequence $\mathbf{X}$.} If the answer is affirmative, Alice says ``yes", and the game ends. 
Otherwise, the game continues, and Alice examines subsequent guesses. In this case, Bob's optimal strategy (if he wishes to minimize the number of guesses) consists of, first, constructing a list of possible password sequences ordered from most to least likely according to $P_X$, and then, querying passwords one by one from this list. 

In this paper, we study a distributed attack scenario, where $m$ adversarial agents receive additional side-information $\mathbf{Y}$ about the password, a so-called target attack \cite{wang2016targeted,wang18codaspy,das2014tangled}. 
In this setting, the agents construct an updated list of password strings, this time, ordered with respect to $P_{X|Y}(\cdot|Y)$, that is they update their belief on the password distribution by taking into account the side-information they have observed.
In its most general form, the side-information can model complex additional information that the adversary may have acquired on the choice of the password, ranging from background search on the user who chooses the password, to behind the back attacks in which an illegitimate person observes parts of the password. 
This setting can also indirectly model adversaries and users over multiple accounts, some of which have been compromised. 
Suppose Alice has several accounts, each requiring a password. 
She may decide to use one identical password for all of her accounts, where the compromise of one of the accounts puts in peril all of her accounts. 
On the other extreme, she may decide to use completely independent passwords for each of the accounts, in which case one password being compromised does not give away any information on any of the other passwords. 
In practice, most users settle for a solution in between these two extremes. 
For example, Alice may choose to slightly tweak her passwords from one account to another as to avoid the disastrous consequences of one account being compromised, while still maintaining some convenience. 
In this case, if one password is compromised, an adversary gains some side-information about the rest of the passwords, see, e.g., \cite{wang18codaspy}.

We say that agents are coordinated if they know the guessing strategies of each other, and in this context it means that the agents are able to communicate about their knowledge of the side-information on the password.
We contrast two strategies illustrated in Fig.~\ref{fig:0}. The first is a decentralized approach in which the agents do not communicate at all, representing the case where agents are fully uncoordinated. The second is a centralized approach in which the side-information is pooled and a central authority provides the optimal lists to the agents, representing a coordinated attack. 
We show that in the case of an uncoordinated attack, having even a finite number of independent sources of side-information reduces the number of queries exponentially. 
However, coordination is very powerful, as complete knowledge of all the side-information can potentially reduce the computational burden on the adversaries by an even bigger exponent. 
This should be contrasted with the case where side-information is unavailable, as the lack of coordination there does not change the computational burden asymptotically.
 
\noindent\textbf{Related Work:} Guesswork has a long history dating from the work of Massey in \cite{massey1994guessing}, where it was first noticed that guesswork is not related to entropy. The problem was picked up again later by Arikan \cite{arikan1996inequality}, in the context of sequential error correcting codes. Since then, Guesswork has been the object of multiple extensions and generalizations, a subset of which is noted below. The problem of a cipher with a guessing wiretapper was considered in \cite{MerhavArikan}. The problem of guessing subject to distortion and constrained Shannon entropy were investigated in \cite{ArikanMerhav} and \cite{Beirami}, respectively. The above results have been generalized to ergodic Markov chains \cite{MaloneSullivan} and a wide range of stationary sources \cite{PfisterSullivan}. The problem of guessing under source uncertainty was investigated in \cite{Sundaresan}. The analysis of the guessing exponents, using large deviations theory, was considered in \cite{HanawalSundaresan}. In \cite{ChristiansenDuffy} it was shown that the guesswork satisfies a large deviation property and the rate function was characterized. Guesswork under erasures was studied in \cite{christiansen2013guessing}, and both the moments and rate functions of the Guesswork were obtained; results which are directly related to our analysis of BEC side-information. A distributed attack model based on password hints was proposed in \cite{bracher2017guessing} and evaluated under guesswork metrics. Applications of Guesswork to botnet attacks were studied in \cite{salamatian2019botnets} and \cite{merhav2018universal}. Subsequently, a geometric characterization of the guesswork was established in \cite{Beirami2}, followed by an information-theoretic characterization of the LDP rate function for guesswork in~\cite{beirami2018characterization}. 
Finally, Salamatian {\it et. al.} provided a  characterization of guesswork using a mismatched distribution~\cite{salamatian2019mismatched}.

\rev{We also mention several papers which are relevant in the study of password generation and brute-force attacks such as \cite{wang2016implications,wang2017zipf,blocki2018economics,bonneau2012science}. In  \cite{bonneau2012science}, a large corpus of password datasets is studied, and several quantities of interest, such as Guesswork, are empirically evaluated. In \cite{wang2016implications,wang2017zipf}, Wang et. al revisit the results of \cite{bonneau2012science}, and provide an improved modeling of the password generation process. In particular, they proposed variants of the Zipf's law distribution (i.e. PDF-Zipf and CDF-Zipf), and validate this model on real-world datasets. The modeling according to a CDF-Zipf's law has also been validated in \cite{blocki2018economics}. We refer to the aforementioned references for more details on the subject.} Specifically related to our setting are targeted attacks, in which the adversary uses the personal information of a user in his guessing strategy, see e.g. \cite{wang2016targeted}. It was shown in \cite{wang2018security,das2014tangled} that these targeted attacks are particularly threatening, as most users chose their passwords according to some personal information which an adversary may have access to (e.g. birthdays, names of family members or pets, locations, or simply password reuse).

\noindent \textbf{Main Contributions:} 
We contrast two strategies, one in which the agents pool their information, referred to as centralized strategy, and one where the agents construct their lists separately, referred to as decentralized strategy. For decentralized strategies we provide a single letter characterization of the asymptotic number of guesses for arbitrary discrete memoryless channels. We complement this with additional and stronger results specialized to the case of binary erasure channel (BEC) and binary symmetric channel (BSC) side-information, where we show that letting the agents agree on a strategy before observing the side-information does not improve the performance. We compare these results with centralized strategies.
Although in both cases, the guesswork is reduced exponentially, our results suggest the strength of cooperation, as two agents sharing their side-information are more powerful than any number of agents acting on their own. 

\noindent \textbf{Previous Publication:} In a conference publication \cite{salamatian2017centralized},
we introduced centralized and decentralized strategies, and studied them for both BEC and BSC channels. \rev{This publication distinguishes itself from our previous work by the following:
	\begin{itemize}
		\item Generalizes the results on decentralized mechanisms to arbitrary memoryless channels (beyond BEC and BSC) in Theorem~3. Note that, this is achieved via a proof technique which is different than the one in \cite{salamatian2017centralized}, as the latter uses a method which is specific to the considered channels. Instead, in this paper, we provide a close-form solution for the guesswork exponent of the decentralized mechanism in terms of quantities which depend on the arbitrary channel $P_{Y|X}$.
		\item Strengthens the results on decentralized mechanisms for the BEC and BSC in \cite{salamatian2017centralized}. More precisely, we show that for the BEC and BSC, planning a joint strategy prior to observing the side-information does not change the value of the guessing exponent (see Remark~1).
		\item Revisits the results in \cite{salamatian2017centralized} into the broader context of the password breaking, and brute-force security literature.
	\end{itemize}
}

\begin{figure*}
	\centering
	\includegraphics[width = .6\textwidth]{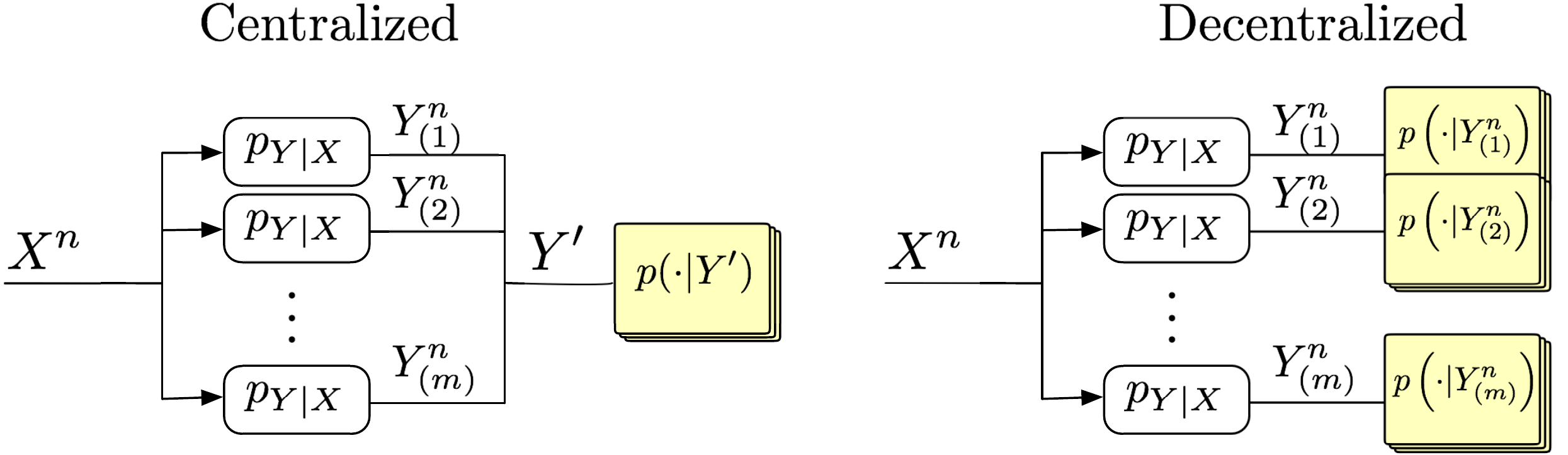}
	\caption{In a coordinated attack, a single list is constructed by collecting all the side-information. In the uncoordinated setting, each agent constructs a separate list.}\label{fig:0}
\end{figure*}

\section{Notation and Background}\label{sec:notation}

Throughout this paper, scalar random variables (RVs) will be denoted by capital letters, and we denote their sample values and alphabets by the corresponding lower case, and calligraphic letters, respectively, e.g. $X$, $x$, and $\mathcal{X}$. 
When considering vectors of random variables, we use the notation $\mathbf{X}_n$ to designate the sequence of RVs $(X_1,\ldots, X_n)$. When clear from the context, we may drop the subscript, e.g., $\mathbf{X}$.
We will reserve the capital letters $P$ and $Q$ to denote probability distribution, which we will subscript with the random variables it is associated with, if applicable. As such, $ (X,Y) \sim Q_{X,Y}$ signifies that the pair $(X,Y) \in \mathcal{X} \times \mathcal{Y}$ is distributed according to the probability distribution $Q_{X,Y}$. As customary, we also denote by channels $X$ to $Y$ by the matrix of single-letter transition $P_{X|Y}$.

The expectation operator over the underlying measure $Q$ will be denoted by $\mathbb{E}_Q\ppp{\cdot}$, and we may once again drop the subscript if clear from the context.
In this paper, we opt to denote information measures by subscripting the corresponding probability distribution, e.g., $I_Q(X;Y)$ shall be the mutual information of $(X,Y) \sim Q_{X,Y}$, $H_Q(X)$ shall be the entropy of $X \sim Q$
Information measures induced by the generic joint distribution $Q_{XY}$, will be subscripted by $Q$, for example, $I_Q(X;Y)$ will denote the corresponding mutual information, etc. The Kullback-Liebler (KL) divergence between two probability measures $P$ and $Q$ will be denoted by $D(P||Q)$. 
The weighted KL divergence between two channels, $Q_{Y|X}$ and $P_{Y|X}$, with weight $P_X$, is defined as
\begin{align}
&D(Q_{Y|X}||P_{Y|X}\vert P_X)\nonumber\\
&\quad\triangleq \sum_{x\in\mathcal{X}}P_X(x)\sum_{y\in\mathcal{Y}}Q_{Y|X}(y\vert x)\log\frac{Q_{Y|X}(y\vert x)}{P_{Y|X}(y\vert x)}. \nonumber
\end{align}
Similarly, for entropies it will be convenient to explicitly write the distributions, e.g. $H(P_X)$, along with the conditional version $H(P_{X|Y}|P_Y)$ defined in the usual way
\begin{align}
	H(P_{Y|X} | P_X) &\triangleq - \sum_{x \in \mathcal{X}} P_X(x)\sum_{y \in \mathcal{Y}} P_{Y|X}(y|x)\log {P_{Y|X}(y|x)} . \nonumber
\end{align}
When dealing with binary random variables we may use the short-hand notation $H(p)$, where it is understood that it refers to the usual entropy over a Bernouilli distribution parametrized by $p$. A similar notation will be used for divergences, e.g., $D(p_1\|p_2)$.

For a given vector $\mathbf{x}_n$, let $\hat{P}_{\mathbf{x}_n}$ denote the empirical distribution, that is, the vector $\{\hat{P}_{\mathbf{x}_n}(x),~x\in{\mathcal{X}}\}$, where $\hat{P}_{\mathbf{x}_n}(x)$ is the relative frequency of the letter $x$ in $\mathbf{x}_n$. 
Let $T(Q_X)$ denote the type class associated with $Q_X$, that is, the set of all sequences $\mathbf{x}_n$ for which $\hat{P}_{\mathbf{x}_n}=Q_X$. 
Similarly, for a pair of vectors $(\mathbf{x}_n,\mathbf{y}_n)$, the empirical joint distribution will be denoted by $\hat{P}_{\mathbf{x}_n,\mathbf{y}_n}$.

The cardinality of a finite set $\mathcal{A}$ will be denoted by $\abs{\mathcal{A}}$, its complement will be denoted by $\mathcal{A}^c$. For any integer $n \in \mathbb{N}^+$, we use the shorthand notation $[1:n] \triangleq \{1, \ldots, n \}$.
The probability of an event $\mathcal{E}$ will be denoted by $\Pr\left\{\mathcal{E}\right\}$. 
For two sequences of positive numbers, $\left\{a_n\right\}$ and $\left\{b_n\right\}$, the notation $a_n\doteq b_n$ means that $\left\{a_n\right\}$ and $\left\{b_n\right\}$ are of the same exponential order, i.e., $n^{-1}\log a_n/b_n\to0$ as $n\to\infty$, where logarithms are defined with respect to (w.r.t.) the natural basis, that is, $\log\left(\cdot\right) = \ln(\cdot)$. 
Finally, for a real number $x$, we denote $[x]_+ \triangleq \max\{0,x\}$.

\noindent \textbf{Background on Guesswork: } We call any one-to-one function $G: \mathcal{X} \to [1:|\mathcal{X}|]$ a guessing function, and let $G(x)$ for $x \in \mathcal{X}$ represent the position of $x$ in the list of guesses induced by $G(\cdot)$. We let $G^*(\cdot)$ be the \emph{optimal} guessing function, obtained by ordering the symbols in $\mathcal{X}$ by decreasing order of $P_X$-probabilities, with ties broken arbitrarily, and letting $G^*(x)$ be the position of $x$ in this list. The problem of bounding the expectation of guesses was investigated in \cite{Arikan}. Specifically, among other things, it was shown \cite[Theorem~1]{Arikan} that for any $\rho\geq0$, and any guessing function $G(\cdot)$,
\begin{align}
\bE\pp{G(X)^\rho}\geq(1+\log\abs{\calX})^{-\rho}\pp{\sum_{x\in\calX}P_X(x)^{\frac{1}{1+\rho}}}^{1+\rho}.
\end{align}
While the optimal guessing function satisfies\footnote{An improved bound by a factor of 2 was reported in \cite{Boztas0}.}
\begin{align}
\bE\pp{G^*(X)^\rho}\leq\pp{\sum_{x\in\calX}P_X(x)^{\frac{1}{1+\rho}}}^{1+\rho}.
\end{align}
Finally, letting $\mathbf{X} = (X_1,X_2,\ldots,X_n)$ be a sequence of independent and identically distributed (i.i.d.) random variables over a finite set, and letting $G^*(\mathbf{X})$ denote the optimal guessing function of a realization of $\mathbf{X}$, it was shown that \cite[Proposition~5]{Arikan}
\begin{align} 
E_\rho(P_X) \triangleq \lim_{n\to\infty}\frac{1}{n}\log\bE\pp{G^*(\mathbf{X})^\rho} = \rho\cdot H_{\frac{1}{1+\rho}}(X_1),\label{ArikanResult}
\end{align}
where $H_\alpha(X)$ is the R\'{e}nyi entropy of order $\alpha$ ($\alpha>0$, $\alpha\neq1$), defined as
\begin{align}
H_{\alpha}(X)\triangleq \frac{1}{1-\alpha}\log\pp{\sum_{x\in\calX}P_X(x)^\alpha}.
\end{align}
Note that the function $E_\rho(P_X)$ simply quantifies the exponential growth of the guesswork, as $n \to \infty$, and thus \eqref{ArikanResult} gives an \emph{asymptotic} operational characterization/meaning to R\'{e}nyi entropy of order $0\leq\alpha\leq1$.

\noindent \textbf{Guesswork with side-information:} We let $Y \in \mathcal{Y}$ be the output of $X$ through a discrete memoryless channel (DMC) with transition probability $P_{Y|X}$. 
Similarly, a guessing function with side information $y$ is any one-to-one function which we denote $G(\cdot|y): \mathcal{X} \to [1:|\mathcal{X}|]$.
Upon receiving a realization $y \in \mathcal{Y}$, Bob updates his belief on the distribution of $X$ by ordering the candidate strings in decreasing order with respect to the posterior $P_{X|Y}(\cdot | y)$. 
We denote by $G^*(\cdot|y)$ the optimal guessing function when the side-information realization is $Y=y$, i.e. $G^*(x|y)$ is the position of $x$ in the optimal list according to the distribution $P_{X|Y}(\cdot|y)$. 
We let the $\rho$-th moment of the \emph{conditional guesswork} $G^*(X|Y)$ be defined as the average:
\begin{align}
	\mathbb{E}\left[G^*(X|Y)^\rho\right] \triangleq \sum_{y \in \mathcal{Y}} P_Y(y) \mathbb{E}[G^*(X|y)^\rho] .
\end{align}
The asymptotic exponent of the conditional guesswork is defined as
\begin{align}
E_\rho(P_X, P_{Y|X}) \triangleq \lim_{n\to\infty} \frac{1}{n} \log \mathbb{E}\left[G^*(\mathbf{X}|\mathbf{Y})\right].
\end{align}
Finally, it was shown in \cite{Arikan} that 
\begin{align}
E_\rho(P_X,P_{Y|X}) &= \rho \cdot H_{\frac{1}{1+\rho}}(X|Y) \\ &= \rho \cdot \sum_{y \in \mathcal{Y}} P_Y(y) H_{\frac{1}{1+\rho}}(X|Y=y).
\end{align}

\section{Coordinated Brute-Force Attack}\label{sec:coordination}

For the remainder of the paper, we assume that a finite number $m$ of sources of side information are available. 
Precisely, for each of the $m$ agents, we consider an independent realization of a side information $\mathbf{Y}_{(i)}, i = 1,\ldots,m$, where $\mathbf{Y}_{(i)}$ is the output of the password sequence $\mathbf{X}$ through a discrete memory-less channel $P_{Y|X}$. It follows that the $\mathbf{Y}_{(i)}$ are identically distributed and independent given $\mathbf{X}$. Recall that coordination refers to the knowledge of the guessing strategies of the other adversaries. Because the optimal guessing strategy of an agent $1 \leq j \leq m$ depends only on the side information $\mathbf{Y}_{(j)}$, coordination is equivalent to sharing the side information. In other words, if no side information is shared, then the adversaries are uncoordinated, and if all the side-information are pooled and shared among all of the $m$ agents, then the adversaries are perfectly coordinated. We consider two strategies the $m$ adversaries may adopt, reflecting two extremes of coordination c.f. Fig.~\ref{fig:0}.

\noindent \textbf{Centralized:} The agents share their observations $\mathbf{Y}_{(i)}$, $i = 1,\ldots,m$, with a central authority which collapses the side information and constructs an optimal list based on $p_{X|Y_{(1)},\ldots,Y_{(m)}}$. The $\rho$-th moment of the guesswork in this strategy is thus,
\begin{equation}
\mathbb{E} \left[ G^*(\mathbf{X} | \mathbf{Y}_{(1)}, \ldots, \mathbf{Y}_{(m)} )^\rho \right],
\end{equation}
where $P_{Y_1,\ldots,Y_m|X}(y_1,\ldots,y_m | x) = \prod_{i = 1}^m P_{Y|X}(y_i|x)$. This corresponds to a completely coordinated attack. Finally, we define,
\begin{align}
	E_{\rho}^{(c)}(P_{Y|X}, m ) \triangleq \lim_{n\to\infty} \frac{1}{n} \log \mathbb{E} \ G^*(\mathbf{X}|\mathbf{Y}_{(1)},\ldots, \mathbf{Y}_{(m)})^\rho  \label{eq:centralized_exp}.
\end{align}

\noindent \textbf{Decentralized Mechanism:} Each of the $m$ agents tries to guess $\mathbf{X}$ based on its own observation $\mathbf{Y}_{(i)}$. The process ends when at least one of the agents correctly guesses $\mathbf{X}$. The $\rho$-th moment of the guesswork for this strategy is thus,
\begin{equation}
\mathbb{E}\left[ \min_{i = 1, \ldots,m} \  G^*(\mathbf{X} | \mathbf{Y}_{(i)})^\rho  \right],
\end{equation}
where $G^*(\mathbf{X}|\mathbf{Y}_{(i)})$ is the optimal guessing function given $\mathbf{Y}_{(i)}$, that is the position of $\mathbf{X}$ in the ordered list according to $P_{\mathbf{X}|\mathbf{Y}}(\cdot|\mathbf{Y} = \mathbf{Y}_{(i)})$. This corresponds to a completely uncoordinated attack. As before, we define
\begin{align}
E_{\rho}^{(d)}(p_{Y|X},m ) \triangleq \lim_{n\to\infty} \frac{1}{n} \log \mathbb{E}\left[ \min_{i = 1, \ldots, m}G^*(\mathbf{X}|\mathbf{Y}_{(i)})^\rho \right] \label{eq:decentralized_exp} .
\end{align}

In the sequel, we shall provide closed-form formulas for \eqref{eq:centralized_exp} and \eqref{eq:decentralized_exp}, and compare them in some examples. It has to be noted that we are studying guesswork behaviors for fixed $m$, that is $m$ may not grow with the block-length $n$. We may take the limit when $m \to \infty$, but it should be clear that the order of limits is crucial and an interchange of limits is not possible here.

\begin{remark}\label{remark:oblivious}
	The decentralized strategy we consider above is one in which each agent produce an optimal list regardless of the list produced by the other agents. In particular, it is not clear that this list should be the joint optimal list strategy. More precisely, it is clear that
	\begin{align}
	\min_{G_{(j)}, j = 1, \ldots, m} & \mathbb{E}\left[\min_{j = 1,\ldots, m} G_j(\mathbf{X}|\mathbf{Y}_{(i)})^\rho \right] \nonumber \\
	& \hspace{4em}\leq \mathbb{E}\left[ \min_{i = 1, \ldots,m} G^*(\mathbf{X} | \mathbf{Y}_{(i)})^\rho \right],
	\end{align}
	but unclear whether equality should hold. While the right-hand side corresponds to an uncoordinated case as we defined it previously, the left-hand side corresponds to a case in which the agents can coordinate in advance to choose their strategies but no more after the side-information is revealed. We shall address this difference when analyzing the performance of the decentralized scheme under some specific side-information channels for which it is possible to characterize the left-hand side, and shall show that they are, at least under these side-information channels, asymptotically identical .
\end{remark}

\rev{To illustrate the centralized and decentralized mechanisms, we consider the following toy example, which is based on the RockYou leaked password dataset.\\
\noindent \textbf{Toy Example:} We extract the top 1000 most likely passwords from the \emph{RockYou} dataset (see \cite{rocku} for a description of the dataset), and limit the scope to passwords with only lowercase letters for convenience. For each such password, we also generate $m = 3$ \emph{sister passwords} synthetically by randomly changing letters, where each letter is changed to any other lower-case letter with a probability of $50\%$. Those sister passwords model the effect of password reuse, and corresponds to the side-information $\mathbf{Y}_{(i)}$, that agent $i = 1, \ldots, m$ has access to. We refer to \cite{wang2016targeted} for an empirical study of the statistics of password reuse, which indicate that many users have a sister password with a small Levenshtein distance.  Examples of passwords along with the synthetic sister passwords are shown in Figure~\ref{tab:passwords}. }

\begin{figure}
	\centering
	\begin{tabular}{|l|l|l|l|l|} 
		\hline
		\multicolumn{1}{|c|}{$\mathbf{X}$} & \multicolumn{1}{c|}{$\mathbf{Y}_{(1)}$} & \multicolumn{1}{c|}{$\mathbf{Y}_{(2)}$} & \multicolumn{1}{c|}{$\mathbf{Y}_{(3)}$} & \multicolumn{1}{c|}{Pooled SI}     \\ 
		\hline
		password                                & wasswgrd               & phssyotd                      & password               & \textcolor[rgb]{0.58,0.067,0}{password}  \\
		iloveyou                                & inoieyou               & izoveyou                      & iloviybv               & \textcolor[rgb]{0.58,0.067,0}{i?oveyou}  \\
		princess                                & prinpess               & pghjcxys                      & wrihness               & \textcolor[rgb]{0.58,0.067,0}{pri??ess}  \\
		rockyou                                 & rockyeu                & rockyou                       & hozkyxu                & \textcolor[rgb]{0.58,0.067,0}{rocky?u}   \\
		nicole                                  & nicoie                 & nbhole                        & zocole                 & \textcolor[rgb]{0.58,0.067,0}{n?cole}    \\
		\hline
	\end{tabular}
\caption{Top 5 lowercase case passwords in the RockYou data. The sister passwords $\mathbf{Y}_{(i)}$ are generated by changing each letter with probability .3 to any lowercase character. The pooled password Side-Information is obtained by taking the letter that appears in more than 50\%  of the sister passwords, and putting an erasure ('?') if no such letter exists..}\label{tab:passwords}
\end{figure}

\rev{
For the sake of exposition, we assume that all letters are equally likely, which is a sub-optimal but illustrative assumption for the purpose of this toy example. Under this assumption, the optimal strategy of an adversary with side information is to modify the sister password one letter at a time, until the correct password is found. Note that, by making use of prior information such as letter frequency, the adversary can improve his guessing strategy drastically -- we refer once again to \cite{wang2016targeted} for an implementation of such guessing strategies. When considering the computational effort (in terms of number of guesses), to recover the password, we can look at two separate scenarios:
\begin{itemize}
		\item A decentralized mechanism, where each agent makes guesses based on its own sister password $\mathbf{Y}_{(i)}$, and the first one to finish determines the computational cost.
		\item A centralized approach, where the sister passwords are pooled. In this case, we assume that any letter that is common in at least 50\% of sister passwords is also in the correct password. Again, this is a sub-optimal guessing strategy, but serves as an illustration. Example of this pooled side-information are shown in Figure~\ref{tab:passwords}.
\end{itemize}}
\rev{In the centralized approach, the quality of the side information is much better, i.e., many of the letters are already correctly recovered, and the remaining sequence to find are only the \emph{erased} symbols. In the decentralized scenario, the side-information is weaker but  there is a benefit in having multiple sources of side-information, as the performance is dominated by the best side-information. The results are showcased in Figure~\ref{fig:exp}, and showcase some of the take-aways from the theoretical analysis to follows. Namely, we see that (1) the presence of sister passwords allows for a greatly reduced computational cost (2) a decentralized approach performs better than a single sister password -- in fact, we will show that this gain is exponential in the analysis that follows, and (3) the centralized approach allows to essentially improve the quality of the side-information, which proves to be a very potent effect. In the rest of the paper, we will show analytically, that for several sources of side-information, a centralized approach with two agents performs asymptotically better than a decentralized approach with any finite number of agents, suggesting that improving the quality of side-information is crucial. 
}

\begin{figure}
	\centering
	\includegraphics[width=.45\textwidth]{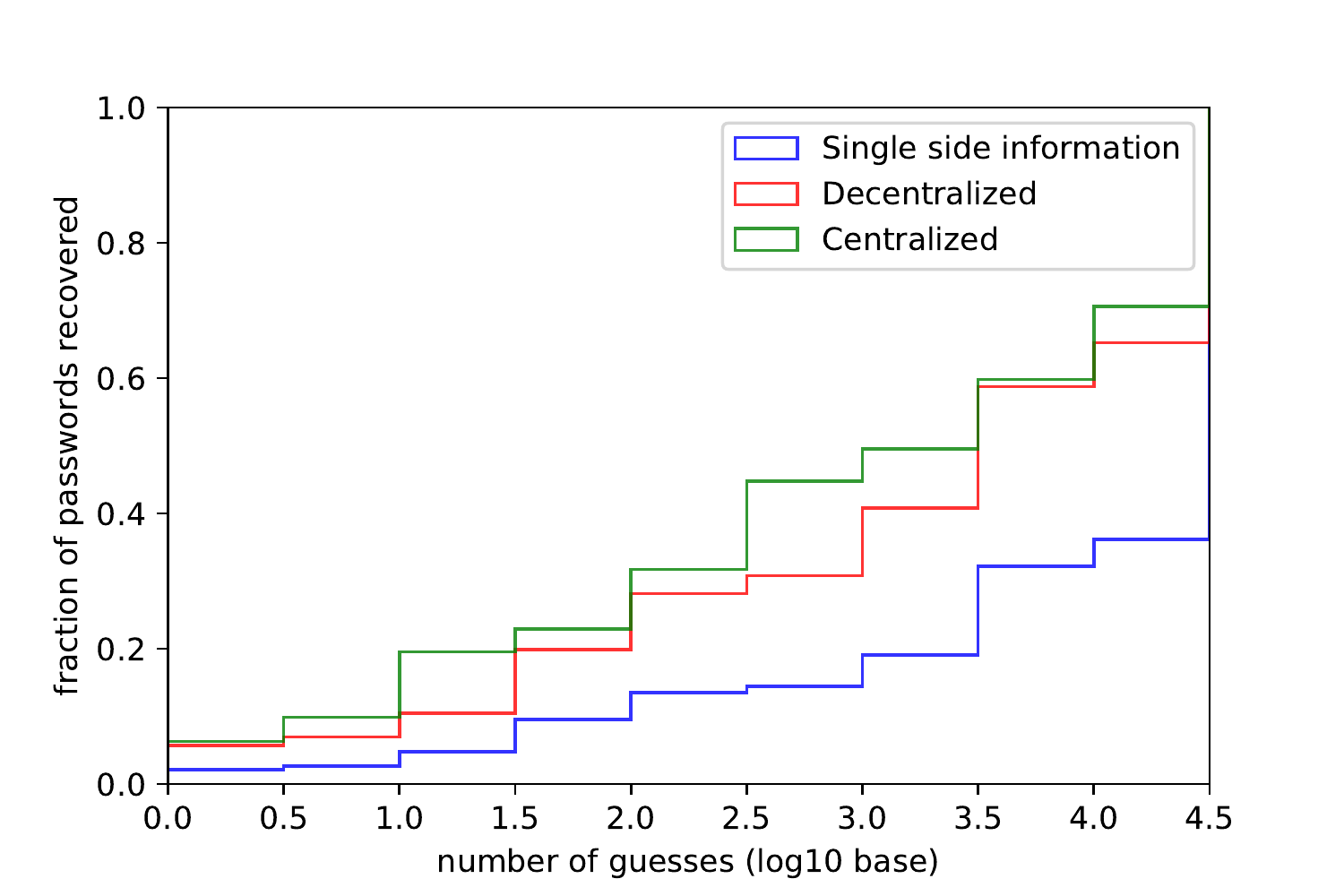}
	\caption{With a centralized mechanism, it takes about 300 guesses to recover 50\% of the passwords. With a decentralized mechanism, it takes several thousand guesses to reach the same performance. Note that an agent with a single side-information, i.e. with a single sister password, recovers only 40\% of the passwords after 30k guesses. }\label{fig:exp}
\end{figure}

\subsection{Centralized Mechanism}\label{sec:centralized}

We illustrate the performance of centralized mechanisms over two side-information channels. First, let $\mathbf{X}$ be a uniformly distributed sequence of binary digits, i.e., $\mathbf{X}$ i.i.d. generated from $\mathrm{Bern}(1/2)$ \footnote{Note that the choice of binary inputs is made for the sake of exposition, and those results can be easily generalized to arbitrary discrete sources.}. We will contrast two types of side-information channels, namely a binary erasure channel (BEC) with parameter $\epsilon$ denoted $\mathrm{BEC}(\epsilon)$, and a binary symmetric channel (BSC) with parameter $\delta$, denoted $\mathrm{BSC}(\delta)$.

We start with the BEC channel. Erasures channels have been studied in \cite{christiansen2013guessing}, where the large deviation principle for the guesswork with erasure side-information was characterized.
This case is simple to analyze because collapsing information is tractable. In particular, the $k$-th entry of $\mathbf{X}$ is erased in all received signals $\mathbf{Y}_i$, $i = 1,\ldots,m$, with probability $\epsilon^m$. Therefore, the resulting collapsed random variable $(\mathbf{Y}_{(1)},\ldots,\mathbf{Y}_{(m)})$ is equivalently described by $\tilde{\mathbf{Y}}$, where $\tilde{\mathbf{Y}}$ is the output of $\mathbf{X}$ through a BEC with erasure probability $\epsilon^m$. We have the following result.

\begin{theorem}[\cite{christiansen2013guessing}]
	For $\mathrm{BEC}(\epsilon)$, and $m$ agents,
	\begin{align}
	E^{(c)}_\rho(\mathrm{BEC}(\epsilon),m) = \max_{\lambda \in [0,1]} \left[ \rho \lambda -  D(\lambda \| \epsilon^m) \right].
	\end{align}
\end{theorem}

Carrying out the maximization for $\rho = 1$, we get the following immediate result.
\begin{corollary}
	For $\rho = 1$,
	\begin{align}
	E^{(c)}_1(\mathrm{BEC}(\epsilon),m) = \log \left( 1 + \epsilon^m \right).
	\end{align}
\end{corollary}

\begin{remark}
	The function $f(x) = \log(1 + x^m)$ over $x \in [0,1]$, is convex for any $m \geq 2$. Moreover, as the number of agents increases, the exponent tends towards a flat function $E^{\text{(c)}}_\rho = 0$, with a discontinuity at $\epsilon = 1$. Finally, since the first derivative (when $\rho = 1$) is $m\frac{\epsilon^{m-1}}{1+\epsilon^{m}}$ for any $m \geq 2$, the centralized curve starts flat with a negligible exponent for small $\epsilon$.
\end{remark}

For the BSC, the centralized mechanism is more involved to analyze. Indeed, we cannot describe the channel resulting from collapsing multiple BSC's in terms of a single BSC anymore, since one has $m$ noisy measurements per password-bit. Nevertheless, for $m = 2$, we can characterize precisely this channel by considering the $2^m = 4$ cases. We shall then discuss how to generalize this result to arbitrary $m > 2$.

\begin{theorem}\label{thm:centralized_BSC}
	For $\mathrm{BSC}(\delta)$, and $m = 2$,
	\begin{align}
	E_\rho^{(c)}(\mathrm{BSC}(\delta),2) = \sup_{\lambda \in [0,1]} &\left\{ \rho\lambda H_{1/1+\rho} \left( \frac{\delta^2}{1 - 2\delta(1-\delta)}\right) \right. \nonumber\\*
	&\left. + \rho (1-\lambda) - D\left(\lambda \| 2\delta(1-\delta) \right) \right\} \nonumber .
	\end{align}
\end{theorem}

\begin{corollary}
	For $\rho = 1$, 
	\begin{align}
	E_1^{(c)}(\mathrm{BSC}(\delta),2) = \log(4\delta(1 - \delta) + 1).
	\end{align}
\end{corollary}

\begin{proof}[Proof of Theorem~\ref{thm:centralized_BSC}]
	Denote by $\mathbf{Y}_{(1)}$ and $\mathbf{Y}_{(2)}$ the sequence of side information observed by each agent. For each bit position, there are two cases: either $\mathbf{Y}_{(1)}$ and $\mathbf{Y}_{(2)}$ agree and have the same value at that position, or they disagree. Without loss of generality, we assume that all agreements appear consecutively with the disagreements following. In the first part, $\mathbf{Y}_{(1)}$ and $\mathbf{Y}_{(2)}$ agree and have the same bit in every position. A simple application of Bayes' rule yields
	\begin{align}
	& P_{X|Y_1,Y_2}( 0 | (0,0) ) = P_{X|Y_1,Y_2}(1|(1,1)) = \frac{(1-\delta)^2}{\delta^2 + (1-\delta)^2} , \\
		&P_{X|Y_1,Y_2}( 1 | (0,0) ) = P_{X|Y_1,Y_2}(0|(1,1)) = \frac{\delta^2}{\delta^2 + (1-\delta)^2} .
	\end{align}
	 that is on this subsequence, the joint side-information $(\mathbf{Y}_{(1)},\mathbf{Y}_{(2)})$ can be equivalently represented by a binary vector $\tilde{\mathbf{Y}}$ which is the result of a BSC with parameter $\delta^2/(1 - 2\delta(1-\delta))$.

	In the second part, $\mathbf{Y}_{(1)}$ and $\mathbf{Y}_{(2)}$ disagree and have contradicting bits in every position. We then have
	\begin{align}
		P_{X|Y_1,Y_2}( 0 | (0,1) ) = P_{X|Y_1,Y_2}(1|(0,1)) = \frac{1}{2},
	\end{align}
	and, 
	\begin{align}
	P_{X|Y_1,Y_2}( 0 | (1,0) ) = P_{X|Y_1,Y_2}(1|(1,0)) = \frac{1}{2},
	\end{align}
	 which is essentially an erasure, since both values of $X$ are equally likely. 
	 We let $\lambda \in [0,1]$ be the fraction of bits over which $\mathbf{Y}_{(1)}$ and $\mathbf{Y}_{(2)}$ agree, \textit{i.e.}, $\lambda n$ is the size of the first subsequence defined above. 
	 Therefore, the central authority has to guess a sequence of the type $\tilde{\mathbf{X}}_n = (\tilde{\mathbf{U}}_{n(1- \lambda)}, \tilde{\mathbf{Z}}_{n \lambda})$, where $\tilde{\mathbf{U}}_{n(1- \lambda)}$ is an i.i.d. sequence of uniform Bernoulli random variables that correspond to the erasures, and $\mathbf{Z}_{n \lambda}$ is an i.i.d. sequence of Bernoulli random variables with parameter $\tilde{\delta} \triangleq \delta^2/(1 - 2\delta(1-\delta))$ which corresponds to the bit-flips. By Lemma~\ref{lem:concatenation_uniform} in the Appendix, we have that
	\begin{align}
	\lim_{n \to \infty} \frac{1}{n} \log \mathbb{E}[G(\tilde{\mathbf{X}})^\alpha] = \lambda \alpha + (1 - \lambda)\alpha H_{1/1+\alpha}(\tilde{\delta}).
	\end{align}
	Noting that the probability of the subsequence of agreements of length $\lambda n$ is (up to polynomial factors) $\exp \left\{-n D(\lambda \| 2\delta(1 - \delta)) \right\}$, we get the desired optimization.
\end{proof}

The previous theorem only treats the case of $m = 2$ agents, although a similar technique can be used to tackle any $m \geq 2$ number of agents. 
Unfortunately, this method is intractable for large $m$. 
However, the following result allows us to compute the limit as the number of agents grows to infinity.
\begin{lemma}
	Assume $\delta \neq \frac{1}{2}$. Then:
	\begin{align}
	\lim_{m \to \infty} E_{\rho}^{(c)} (\mathrm{BSC}(\delta),m) = 0.
	\end{align}
\end{lemma}
\begin{proof}
    Without loss of generality, let $\delta < 1/2$.
	For a fixed $n$ and $m$, we do a deterministic pre-processing on the sequences $\mathbf{Y}_{(1)}, \ldots, \mathbf{Y}_{(m)}$, which can only increase the guesswork, by definition. 
	We let $\hat{Y}_k$ be defined as the majority bit among the received side information sequences at index $k$, that is,
	\begin{align}
	\hat{Y}_k\triangleq \left\{
	\begin{array}{ll}
	0, & \text{if } N_k(0) \geq N_k(1) \\
	1, & \text{if } N_k(0) < N_k(1)
	\end{array} \right.
	\end{align} 
	where $N_k(0) \triangleq \sum_{j = 1}^m Y_{(j),k}$, $Y_{(j),k}$ is the $k$-th bit of the sequence $\mathbf{Y}_{(j)}$, and $N_k(1) \triangleq n - N_k(0)$. 
	Then, it is easy to see that the sequence $\hat{\mathbf{Y}} \triangleq (\hat{Y}_1,\ldots,\hat{Y}_n)$ is the output of $\mathbf{X}$ through a BSC with parameter $\delta_m$, such that $\delta_m \to 0$ as $m \to \infty$, for any $\delta < 1/2$ \footnote{A bound on $\delta_m$ can be obtained by an application of Chernoff bound, i.e., $\delta_m < e^{-nD(1/2\| \delta)}$}. Therefore, for any $n$ and, fixed $m$, the following equations hold:
	\begin{align}
	\mathbb{E}[G(\mathbf{X}|\tilde{\mathbf{Y}})^\rho] &\leq \mathbb{E}[G(\mathbf{X}|\hat{\mathbf{Y}})^\rho] ,\\
	\implies E_\rho^{\text{(c)}} (\mathrm{BSC}(\delta),m) &\leq E_\rho(\mathrm{BSC}(\delta_m)) ,\\
	\implies \lim_{m \to \infty} E_\rho^{\text{(c)}} (\mathrm{BSC}(\delta),m) &\leq \lim_{m \to \infty} E_\rho(\mathrm{BSC}(\delta_m)).
	\end{align}
	Since the right hand side of the last inequality converges to $0$, for any $\delta < \frac{1}{2}$, we obtain the desired result.
\end{proof}
In other words, when $m$ is large enough, one can \emph{estimate} each bit of the password based on the noisy observations.

\subsection{Decentralized Mechanism}\label{sec:decentralized}

We now study the number of guesses per adversary under the decentralized approach. Our main result, presented below, gives an asymptotic single letter formula for \eqref{eq:decentralized_exp}.

\begin{theorem}\label{thm:decentralized}
	Let $\mathbf{X}$ be generated i.i.d. from $P_X$. Then,
	\begin{align}
	&\lim_{n\to\infty} \frac{1}{n} \log \mathbb{E}\left[\min_{i = 1, \ldots,m} G(\mathbf{X}|\mathbf{Y}_{(i)})^\rho \right] =  \nonumber \\*
	 & \sup_{\alpha \in [0, 1]} \hspace{-5em}\sup_{\subalign{\hspace{5.5em} & \hat{P}_{X,Y} \\ & \text{subject to} \; \hat{P}_{X|Y} \notin \mathcal{Q}\left(\alpha,\hat{P}_Y \right)}}\hspace{-5em}  \rho \cdot \alpha - D(\hat{P}_X || P_X) - m D(\hat{P}_{Y|X} || P_{Y|X} | \hat{P}_X) \label{eq:decentralized_closedform}
	\end{align}
	where $\mathcal{Q}(\alpha,\hat{P}_{\mathbf{y}})$ is defined as
	\begin{align*}
	\mathcal{Q}(\alpha,\hat{P}_Y) \triangleq \left\{ Q_{X|Y} : D(Q_{X|Y}\| P_{X|Y}|\hat{P}_Y) + H(Q_{X|Y}|\hat{P}_Y) \right. \\
	\left. <D(Q^*_{X|Y}\| P_{X|Y}|\hat{P}_Y) + H(Q^*_{X|Y}|\hat{P}_Y) \right\},
	\end{align*}
	with $Q^*_{X|Y}$ being the solution of the optimization problem
	\begin{equation}\label{eq:threshold_prob}
	\begin{aligned}
	& \underset{Q_{X|Y}}{\text{minimize}}
	& & D(Q_{X|Y} \| P_{X|Y} | \hat{P}_Y) + H(Q_{X|Y}|\hat{P}_Y) \\
	& \text{subject to}
	& & H(Q_{X|Y}|\hat{P}_Y) \geq \alpha .
	\end{aligned}
	\end{equation}
\end{theorem}

\begin{proof}[Proof of Theorem~\ref{thm:decentralized}]
	We consider the case of $\rho = 1$. The generalization for any $\rho \geq 0$ is immediate. We start by conditioning on $\mathbf{X}$,
	\begin{align}
		\mathbb{E}\left[\min_{j = 1, \ldots, m} \hspace{-.75em} G^*\left(\mathbf{X}|\mathbf{Y}_{(j)} \right)\right] = \mathbb{E} \left[\mathbb{E}\left[\min_{j = 1, \ldots, m} \hspace{-.75em}G^*(\mathbf{X}|\mathbf{Y}_{(j)} ) \middle| \mathbf{X} \right]\right] . \label{eq:conditioning}
	\end{align}
	Since $\min_{j = 1, \ldots, m} G^*(\mathbf{X}|\mathbf{Y})$ is non-negative, \rev{and recalling that $\mathbb{E}[X] = \sum_{i \geq 0}\mathbb{P}(X \geq i)$ for a non-negative random variable $X$}, we have that the inner expectation on the right hand side evaluates to
	\begin{align}
 \sum_{i = 1}^{|\mathcal{X}|^n} \mathbb{P} \left\{  \min_{j = 1,\ldots, m} G^*\left(\mathbf{X} | \mathbf{Y}_{(j)} \right) \geq i \middle| \mathbf{X} = \mathbf{x}  \right\} . \label{eq:sum_i}
	\end{align}
	For a fixed $i$ and $\mathbf{x} \in \mathcal{X}^n$, note that $\mathbf{Y}_{(j)}$ are independent given $\mathbf{X}$, and thus $G^*(\mathbf{X}|\mathbf{Y}_{(j)})$ are independent and identically distributed given $\mathbf{X}$. We then have
	\begin{align}
	& \mathbb{P} \left\{  \min_{j = 1,\ldots, m} G^*\left(\mathbf{X} | \mathbf{Y}_{(j)} \right) \geq i \middle| \mathbf{X} = \mathbf{x}  \right\} \\
	& \hspace{2em}= \prod_{j= 1}^{m} 	\mathbb{P} \left\{  G^*\left(\mathbf{X} | \mathbf{Y}_{(j)} \right) \geq i \middle| \mathbf{X} = \mathbf{x}  \right\} \label{eq:ind_min}\\
	& \hspace{2em} = 	\left[ \mathbb{P} \left\{ G^*\left(\mathbf{X} | \mathbf{Y}_{(1)} \right) \geq i \middle| \mathbf{X} = \mathbf{x}  \right\} \right]^m ,
	\end{align}
	\rev{where we have used independence in \eqref{eq:ind_min}.}
	Next, we have,
	\begin{align}
		& \mathbb{P} \left\{ G^*\left(\mathbf{X} | \mathbf{Y}_{(1)} \right) \geq i \middle| \mathbf{X} = \mathbf{x}  \right\}  \\ 
		& \hspace{.25em}= \sum_{\mathbf{y} : G^*(\mathbf{x}|\mathbf{y}) \geq i } P_{\mathbf{Y}|\mathbf{X}}(\mathbf{y} | \mathbf{x} ) \\
		& \hspace{.25em} = \sum_{\mathbf{y} \in \mathcal{L}_i(\mathbf{x})} \hspace{-.75em} \exp \ppp{  -n \left[D\left(\hat{P}_{\mathbf{y}|\mathbf{x}} \| P_{Y|X} \middle| \hat{P}_{\mathbf{x}} \right) + H\left(\hat{P}_{\mathbf{y}|\mathbf{x}} \middle| \hat{P}_{\mathbf{x}} \right) \right] }, \label{eq:types_1}
	\end{align}
	where $\mathcal{L}_i(\mathbf{x})$ corresponds to the set $\mathcal{L}_i(\mathbf{x}) \triangleq \left\{\mathbf{y} \in \mathcal{Y}^n : G(\mathbf{x}|\mathbf{y}) \geq i  \right\}$, and $\hat{P}_{\mathbf{x}}$ and $\hat{P}_{\mathbf{y}|\mathbf{x}}$ correspond to the empirical distribution (type) of $\mathbf{x}$, and $\mathbf{y}$ given $\mathbf{x}$, respectively (see \cite[Lemma~2.6]{CsisKro}). A given sequence $\mathbf{y}$ with conditional type $\hat{P}_{\mathbf{x}|\mathbf{y}}$ induces a \emph{reverse channel} $\hat{P}_{\mathbf{x}|\mathbf{y}} = \frac{\hat{P}_{\mathbf{x}|\mathbf{y}} \hat{P}_{\mathbf{x}}}{\hat{P}_\mathbf{y}}$. The condition $\mathbf{y} \in \mathcal{L}_i(\mathbf{x})$ can then be expressed in terms of this reverse channel, as the position of $\mathbf{x}$ in the optimal list constructed according to $P_{X|Y}$ is essentially a function of the types $\hat{P}_\mathbf{y}$ and $\hat{P}_{\mathbf{x}|\mathbf{y}}$, and the value of $\alpha \triangleq \log i$, as shown in Lemma~\ref{lem:types_si} in the Appendix. Thus, using the method of types \cite[Chapter~2]{CsisKro}, we may rewrite \eqref{eq:types_1} as follows
	\begin{align}
		&\mathbb{P}  \left\{ G^*\left(\mathbf{X} | \mathbf{Y}_{(1)} \right) \geq i \middle| \mathbf{X} = \mathbf{x}  \right\}  \\
		 & \hspace{.5em}= \hspace{-2em}\sum_{\hat{P}_{\mathbf{x},\mathbf{y}} \notin \mathcal{Q}\left(\alpha,\hat{P}_Y \right)} \hspace{-2em} \left| T(\hat{P}_{\mathbf{y}|\mathbf{x}})\right| \exp \left\{  -n \left[D\left(\hat{P}_{\mathbf{y}|\mathbf{x}} \| P_{Y|X} \middle| \hat{P}_{\mathbf{x}} \right) \right. \right. \nonumber \\
		 & \hspace{12em} \left.  \left. + H\left(\hat{P}_{\mathbf{y}|\mathbf{x}} \middle| \hat{P}_{\mathbf{x}} \right) \right] \right\}  \\
		&\hspace{.5em} \de  \hspace{-5em}\sum_{\subalign{\hspace{5.5em} & \hat{P}_{X,Y} \\ & \text{subject to} \; \hat{P}_{X|Y} \notin \mathcal{Q}\left(\alpha,\hat{P}_Y \right)}} \hspace{-5em}  \exp \left\{  -n \left(D\left(\hat{P}_{\mathbf{y}|\mathbf{x}} \| P_{Y|X} \middle| \hat{P}_{\mathbf{x}} \right)\right) \right\} \\
		&\hspace{.5em} \de \hspace{-5em}\sup_{\subalign{\hspace{5.5em} & \hat{P}_{X,Y} \\ & \text{subject to} \; \hat{P}_{X|Y} \notin \mathcal{Q}\left(\alpha,\hat{P}_Y \right)}} \hspace{-5em}   \exp \left\{  -n \left(D\left(\hat{P}_{\mathbf{y}|\mathbf{x}} \| P_{Y|X} \middle| \hat{P}_{\mathbf{x}} \right)\right) \right\} . \label{eq:types_2}
	\end{align}
	We are now ready to plug \eqref{eq:types_2} into \eqref{eq:sum_i}. Recall that the position of $\mathbf{x}$ is a function of the types $\hat{P}_{\mathbf{x}|\mathbf{y}}$ and $\hat{P}_\mathbf{y}$. Let the set $\mathcal{A} = \{\alpha :  \alpha = H(\hat{P}_{\mathbf{x}}), \text{for some sequence } \mathbf{x} \in \mathcal{X}^n\}$, be the set of empirical entropy values which can be obtained from the $n$-length sequences. Note that since there are only a polynomial number, in $n$, of valid types $\hat{P}_{\mathbf{x}}$, $\mathcal{A}$ is also of polynomial size, and thus, we have
	\begin{align}
	&\sum_{i = 1}^{|\mathcal{X}|^n} \mathbb{P}  \left\{  \min_{j = 1,\ldots, m} G^*\left(\mathbf{X} | \mathbf{Y}_{(j)} \right) \geq i \middle| \mathbf{X} = \mathbf{x}  \right\}  \nonumber \\
	& \hspace{.5em}= \sum_{\alpha \in \mathcal{A}}e^{n\alpha} \mathbb{P} \left\{  \min_{j = 1,\ldots, m} G^*\left(\mathbf{X} | \mathbf{Y}_{(j)} \right) \geq \lceil |\mathcal{X}|^{n\alpha} \rceil \middle| \mathbf{X} = \mathbf{x}  \right\}  \nonumber \\
	& \hspace{.5em}\de \sup_{\alpha \in [0, 1]}  \hspace{-5em}\sup_{\subalign{\hspace{5.5em} & \hat{P}_{X,Y} \\ & \text{subject to} \; \hat{P}_{X|Y} \notin \mathcal{Q}\left(\alpha,\hat{P}_Y \right)}} \hspace{-5em} \exp \left\{  n \left[ \alpha - D\left(\hat{P}_{\mathbf{y}|\mathbf{x}} \| P_{Y|X} \middle| \hat{P}_{\mathbf{x}} \right)\right] \right\} \label{eq:types_3} .
	\end{align}
	Finally, plugging \eqref{eq:types_3} into \eqref{eq:conditioning}, and  using once again the method of types to get that $\mathbb{P}(\mathbf{X} \in T(\hat{P}_{\mathbf{x}})) \de \exp\{-nD(\hat{P}_{\mathbf{x}} \| P_X )\}$, the result is deduced.
\end{proof}

Using Theorem~\ref{thm:decentralized}, we have the following corollary.
\begin{corollary}\label{cor:decentralized}
	For any $\rho > 0$,
	\begin{align}
		\lim_{m\to\infty} 	\lim_{n\to\infty} \frac{1}{n} \log \mathbb{E}\left[\min_{i = 1, \ldots,m} \left\{ G(\mathbf{X}|\mathbf{Y}_{(i)})^\rho \right\} \right] = H(X|Y) .
	\end{align}
\end{corollary}
\begin{proof}
	Looking at \eqref{eq:decentralized_closedform}, we see that as $m \to \infty$, $D(\hat{P}_{Y|X} \| P_{Y|X} | \hat{P}_X)$ must be zero, and thus $\hat{P}_{Y|X}$ must be equal to $P_{Y|X}$ for all $x$ such $\hat{P}_X(x)> 0$. Note that, the maximizing $\hat{P}_X$ is then given by $\hat{P}_X = P_X$, and thus we get $\hat{P}_{X,Y} = P_{X,Y}$. This in turns impose a condition on $\alpha$, namely that the set $\mathcal{Q}(\alpha,P_Y)$ must not contain $P_{X|Y}$. Precisely, we have
	\begin{align}
		P_{X|Y} \notin \mathcal{Q}(\alpha|P_Y) \implies H(P_{X|Y}|P_Y) \geq H(Q^*_{X|Y}|P_Y) \geq \alpha,
	\end{align}
	where the second inequality follows from the definition of $Q^*_{X|Y}$. Thus, the maximal $\alpha$ is given by $H(P_{X|Y}|P_Y) = H(X|Y)$.
\end{proof}

To illustrate the power of the decentralized approach, we consider again the BEC and BSC side information. Note that it is possible to obtain these results by plugging in Theorem~\ref{thm:decentralized}. However, for these two channels, it is insightful to take a direct approach. In addition, we address Remark~\ref{remark:oblivious}, and show that under these two channels, the number of guesses does not change asymptotically even if the adversaries coordinate jointly their lists prior to observing the side-information.

\begin{theorem}\label{thm:distrib_BEC}
For $\mathrm{BEC}(\epsilon)$,
	\begin{align}
	E_\rho^{(d)}(\mathrm{BEC}(\epsilon),m)= \sup_{\lambda \in [0,1]} \left( \rho \lambda - m D(\lambda || \epsilon) \right).
	\end{align}
\end{theorem}

Before we proceed to the proof of Theorem~\ref{thm:distrib_BEC}, some remarks are in order. One can verify that the guesswork exponent for the decentralized mechanism, as the number of agents $m$ increases, converges towards $\epsilon$ (see Fig~\ref{fig:1}), as expected from Corollary~\ref{cor:decentralized}. On the other hand, Remark~1 implies that even two agents that collapse their side information are more powerful than any finite number of agents guessing $\mathbf{X}$ in a decentralized way, since the centralized scheme has a convex exponent.

\begin{proof}[Proof of Theorem~\ref{thm:distrib_BEC}]
	For simplicity of exposition, we focus on the case where $m = 2$ and $\rho = 1$, while the generalization for any $\rho$ and $m$ is immediate. The proof of Theorem~\ref{thm:distrib_BEC} follows from two steps. 
	First, we find an upper bound on the guesswork exponent by considering the exponent of the shortest sequence \footnote{Note that this exponent can be derived directly as a consequence of the results in \cite{christiansen2013guessing}. The proof method in this paper is included for completeness, and characterizes only the exponent of the guesswork, as opposed to the entire large deviation rate, as done in \cite{christiansen2013guessing}.}. Recall that, since $\mathbf{Y}_{(i)}$ is just an erased version of $\mathbf{X}$ for any $i=1,\ldots,m$. The adversaries must each try to guess a sequence $\mathbf{Z}_{(i)}$, where the length of $\mathbf{Z}_{(i)}$ is the number of erasures in the sequence $\mathbf{Y}_{(i)}$, denoted $\mathcal{E}_n^{(i)}$, and $\mathbf{Z}_{(i)}$ is a uniformly distributed binary sequence. We then have
	\begin{align}
	\mathbb{E}[\min_{i = 1,\ldots,m} G(\mathbf{Z}_{(i)})] \leq \mathbb{E}[G(\mathbf{Z}_{*})],
	\end{align}
	where $\mathbf{Z}_*$ is the sequence of any adversary which has $\mathcal{E}^*_n \triangleq \min_{i = 1} \mathcal{E}_n^{(i)}$ erasures. Note that the probability of having $\mathcal{E}^*_n = n \cdot \lambda$ for some $\epsilon \leq \lambda \leq 1$, is given (exponentially) by $\exp\left[-n \cdot mD(\lambda \| \epsilon)\right]$. Indeed, 
	\begin{align}
	\mathbb{P}\left(\frac{1}{n}\mathcal{E}_n^* =\lambda \right) & \doteq \mathbb{P}\left(\frac{1}{n}\mathcal{E}_n^* \leq\lambda\right) \\
	& = \mathbb{P}\left(\frac{1}{n}\mathcal{E}_i \leq\lambda\right)^m \\
	& \doteq \exp \left[-n \cdot m D(\lambda \| \epsilon) \right],
    \end{align}
    where the last step follows from Sanov's theorem. Similarly, when $0 \leq \lambda < \epsilon$, the probability $ \mathbb{P}\left(\frac{1}{n}\mathcal{E}_n^* =\lambda \right) $ is exponentially equal to $\exp - n D(\lambda \| \epsilon)$. Therefore, letting $f(\lambda,m) = \mathbf{1}\left\{ \lambda > \epsilon \right\} m D(\lambda \| \epsilon) + \mathbf{1}\left\{ \lambda \leq \epsilon \right\} D(\lambda \| \epsilon)$, we have:
	\begin{align}
		\mathbb{E}\left[G(\mathbf{Z}_*)\right] &= \mathbb{E}\left[\mathbb{E}\left[G(\mathbf{Z}_*) | \mathcal{E}_* = n \lambda \right]\right] \\
		& = \sum_{\lambda = 0, 1/n, \ldots, 1} \mathbb{P}(\mathcal{E}_* = \lambda n) \exp(n \lambda) \\
		& \de \exp \left[ n \sup_{\lambda \in [0,1]} \left( \lambda - f(\lambda, \epsilon) \right) \right].
	\end{align}
	Noting that the maximizing $\lambda$ is always greater or equal to $\epsilon$, we have the upper-bound
	\begin{align}
			\lim_{n \to \infty} \frac{1}{n} \log \mathbb{E}\left[ \min_{i = 1,\ldots,m} G(\mathbf{Z}_{(i)}) \right] \leq \sup_{\lambda \in [0,1]} \left[ \lambda - mD(\lambda \| \epsilon) \right].
	\end{align}
	To obtain a matching lower-bound, we consider an oracle that provides additional information to both agents, strictly reducing their guesswork. The additional information from the oracle allows to construct explicitly the optimal list of both agents. More precisely, this is achieved by transmitting the position of the common erasures for both agent. The optimal joint strategy is then to construct lists as to minimize queries that have a common subsequence in the overlapping erasures. Indeed, each incorrect query from an agent, shapes the probability distribution of the second agent because of the common sequences. We show that this probability shaping, can be again lower-bounded by a mechanism in which each agent has two guesses at each step, instead of one, therefore not affecting the guesswork exponent. This is formalized below:
	\begin{figure}
	    \centering
	    \includegraphics[scale = .4]{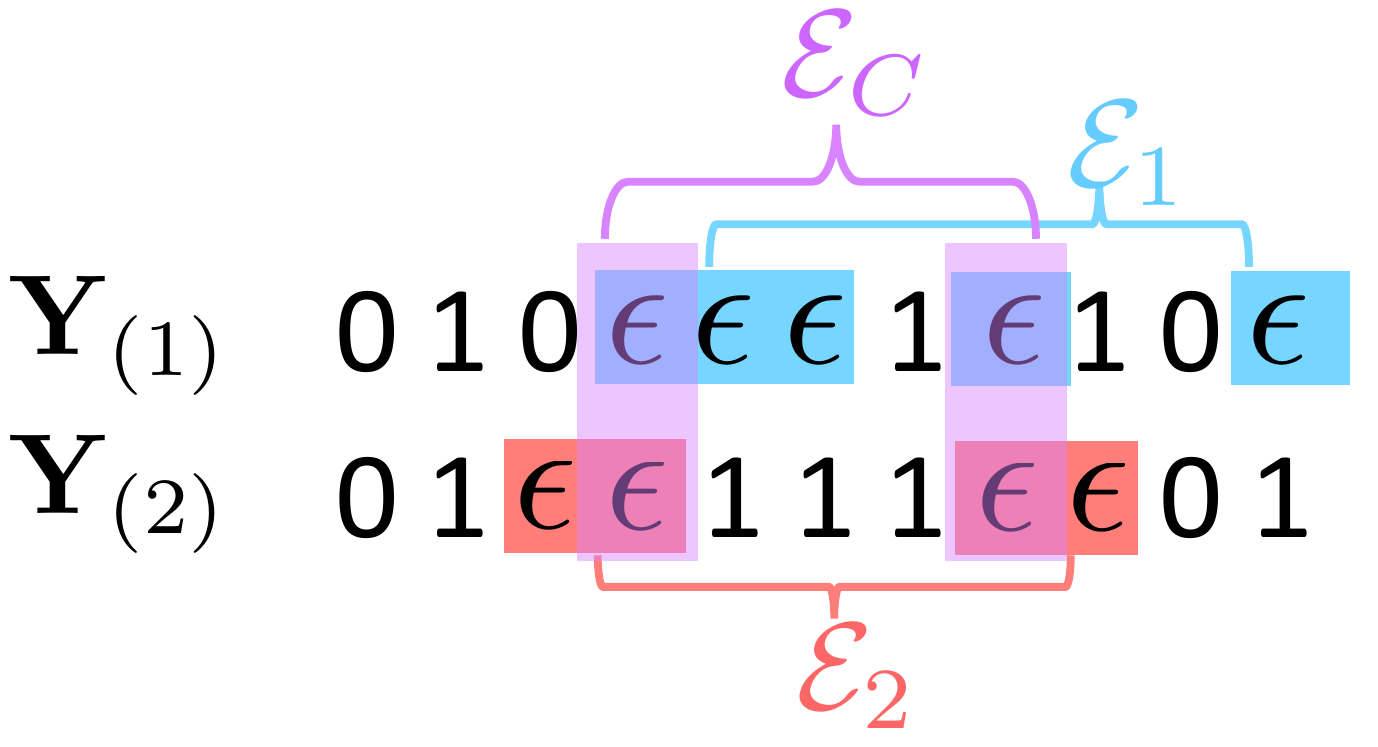}
	    \caption{The erasures sets that the Oracle Mechanism shares. Note that $G^*(\mathbf{X}|\mathbf{Y}_{(1)})$ and $G^*(\mathbf{X}|\mathbf{Y}_{(2)})$ are not independent because of the bits in $\mathcal{E}_C$. Over this interval, the agents should query sequences which are disjoint, by for example, querying following opposite ends of a lexicographical ordering.}
	    \label{fig:oracle}
	\end{figure}
	\begin{definition}[Oracle Mechanism] 
		Let $\mathcal{E}_1$ be the set of erased indices for agent 1, that is $\mathcal{E}_1 = \{i | Y_{(1),i} = \epsilon \}$, and define $\mathcal{E}_2$ similarly for agent 2. Also let $\mathcal{E}_C = \mathcal{E}_1 \cap \mathcal{E}_2$ be the common erasures, and denote by $n_c = |\mathcal{E}_C|$. Further, let $n_1 = |\mathcal{E}_1 \backslash \mathcal{E}_C|$ and $n_2 = |\mathcal{E}_2 \backslash \mathcal{E}_C|$, see Fig.~\ref{fig:oracle}. Suppose without loss of generality that $n_1 \geq n_2$ .We consider an helping oracle that does the following:
		\begin{itemize}
			\item Transmits to each agent the sets $\mathcal{E}_1$ and $\mathcal{E}_2$. 
			\item Reveals $n_1 - n_2$ bits among those in $\mathcal{E}_1 \backslash \mathcal{E}_C$ to agent 1, making agent 1 as strong as agent 2.
		\end{itemize}
		That is, agent $i$ has to guess a binary uniform sequence $(\tilde{X}_{(i)}^{n_1}, \tilde{X}^{n_c})$, where the subsequence $\tilde{X}^{n_c}$ is common for both agents, and the subsequences $\tilde{X}^{n_1}_{(1)}$ and $\tilde{X}^{n_1}_{(2)}$ are independent.
	\end{definition}
	With the knowledge of the Oracle, the two agents will try to construct an optimal joint strategy. At step $k$, the agent $1$ will pick its sequence assuming its previous $k-1$ were incorrect, as well as the $k-1$ sequences of the second agent. Indeed, each of the $k-1$ guesses of the second agent shapes the probability distribution over the sequences for the first agent due to the common sequence $\tilde{X}^{n_c}$. Therefore, the optimal strategy for the agent $1$ is to query a sequence for which the corresponding subsequence $\tilde{x}^{n_c}$ is as likely as possible, or in other words, has been queried the least so far by the other agent. This can be achieved simply by considering a lexicographical ordering over the subsequences $x^{n_c}$ for one agent and an anti-lexicographical ordering for the other agent, as this guarantees that each agent queries sequences that disagree on their subsequence. Next, using the Lemma~\ref{lem:soft_elim}, we show that this process is worse, in terms of guesswork, to a process in which the agent gets one \emph{free} query. Therefore, the guesswork is unchanged asymptotically, and we obtain the desired result.
\end{proof}

We now study the BSC side-information channel.

\begin{theorem}\label{thm:distrib_BSC}
	For $\mathrm{BSC}(\delta)$,
	\begin{align}
	E_\rho^{(d)}(\mathrm{BSC}(\delta), m) = \rho H_{\frac{m}{\rho+m}}(\delta).
	\end{align}
\end{theorem}
\begin{proof}[Proof of Theorem~\ref{thm:distrib_BSC}]
	First notice that $\mathbf{Y}_{(i)} = \mathbf{X} \oplus \mathbf{Z}_{(i)}$ where $\mathbf{Z}_{(i)}$ is the sequence of flips, and is generated i.i.d. from $\mathrm{Bern}(\delta)$, and hence $G(\mathbf{X}|\mathbf{Y}_{(i)}) = G(\mathbf{Z}_{(i)})$. Further, all $\mathbf{Z}_{(i)}$ sequences are independent, and so are the guessworks $G(\mathbf{Z}_{(i)})$. First recall the following elementary result. Let $S^n_i$, for $i = 1,\ldots,m$, be the sum of $n$ i.i.d. coin flips with parameter $\delta$, and let $S_1^n, \ldots, S_m^n$ be independent. Then, for any $\delta < s \leq 1$:
	\begin{align*}
	&\mathbb{P} \left(\min_{i} S_i = sn \right) = m \cdot \mathbb{P}(S_1 = s\cdot n) \cdot \prod_{i=2}^m \mathbb{P}(S_i \geq s\cdot n) \\
	 &\quad\quad\quad\quad\quad\quad\de \exp \{- n D(s || \delta)  \} \left( \exp\{ -nD(s||\delta)\} \right)^{m-1} \\
	 &\quad\quad\quad\quad\quad\quad\de \exp \{ -nm D(s||\delta)\}.
	\end{align*}
	Alternatively, when $0 < s \leq \delta$, we have:
	\begin{align}
	\mathbb{P} \left(\min_{i = 1,\ldots,m} S_i = sn\right) & \de \exp\{ -nD(s||\delta)\}.
	\end{align}
	Using the previous results, and recalling that $G(Z^n_{(i)}) \de 2^{S^n_i}$, where $S^n_i$ is the number of 0's in the sequence (the type of the binary sequence), we obtain that:
	\begin{align}
	\mathbb{E} \left[\min_{i = 1,\ldots,m} G(\mathbf{Z}_{(i)})^\rho \right] \de \exp \left\{ n \cdot \sup_{\lambda \in [0,1]} \left(  \rho \lambda - f(\lambda,m) \right) \right\},
	\end{align}
	where $f(\lambda,m) = \mathbf{1}\{ \lambda > \delta \} mD(\lambda ||\delta) + \mathbf{1}\{\lambda \leq \delta \} D(\lambda || \delta)$. The desired result follows by observing that the maximization over $\lambda$ always lead to a solution in the range $\lambda > \delta$, for any $\rho > 0$.
\end{proof}

The results of Theorems \ref{thm:centralized_BSC},\ref{thm:distrib_BEC} and \ref{thm:distrib_BSC} are illustrated in Figures~\ref{fig:1} and \ref{fig:2}.

\begin{figure}
		\centering
	\begin{minipage}{.45\textwidth}
	\centering
	\includegraphics[scale=.52]{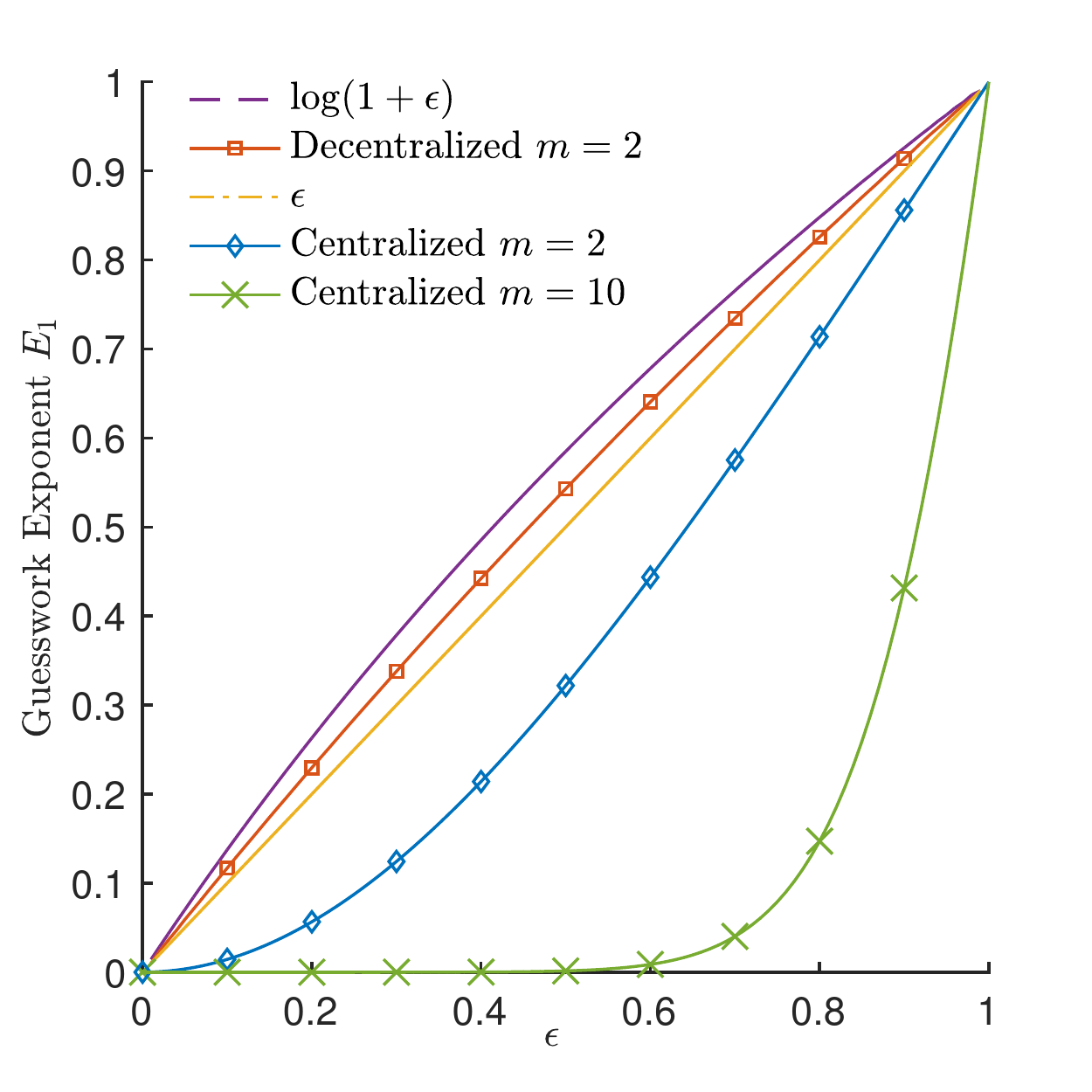}
	\caption{$\mathrm{BEC}(\epsilon)$: Exponents of the average guesswork (i.e. $\rho = 1$) for various $m$, and under centralized and decentralized strategies. Note that two cooperating agents have a convex exponent, which is better than any number of non-cooperating agents.}\label{fig:1}
	\end{minipage}
\quad
	\begin{minipage}{.45\textwidth}
	\centering
	\includegraphics[scale = .52]{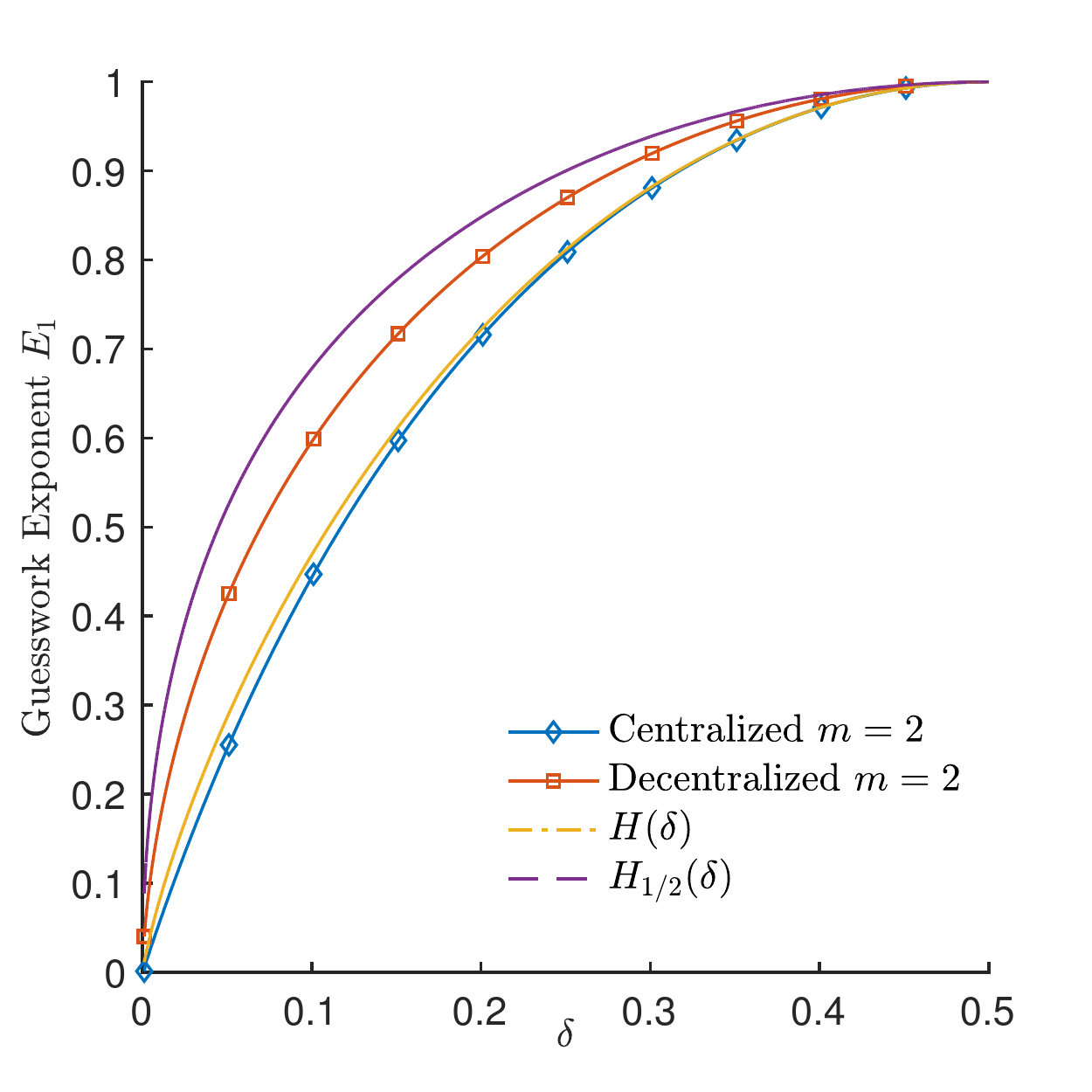}
	\caption{$\mathrm{BSC}(\delta)$: Exponents of the average guesswork (i.e. $\rho = 1$) for $m=2$ and as $m \to \infty$, and under centralized and decentralized strategies. Again, two cooperating agents have a better exponent than any number of non-cooperating agents.}\label{fig:2}
	\end{minipage}

\end{figure}

\appendices
\numberwithin{equation}{section}

\section{Additional Lemmas}
The following lemma characterizes the guesswork exponent of a sequence generated by the concatenation of a uniform binary sequence, and an arbitrary \emph{i.i.d.} sequence.
\begin{lemma}\label{lem:concatenation_uniform}
	Let $U \sim \mathrm{Bern}(1/2)$ and $V \sim \mathrm{Bern}(p)$, with $p \leq 1/2$, and denote by $U^{m_n}$ and $V^{n-m_n}$ their \emph{i.i.d.} sequences, for some sequence $m_n$ such that $\lim_{n \to \infty} \frac{m_n}{n} = \lambda$. Then, the guesswork exponent for sequence $X^n = (U^{m_n}, V^{n - m_n})$ is:
	\begin{align}
	\lim_{n \to \infty} \log \mathbb{E} \left[ G(\mathbf{X})^\rho \right] = \lambda \rho + (1 - \lambda) \rho H_{1/1+ \rho} (p).
	\end{align}
\end{lemma}
\begin{proof}
	We do the proof for $\rho = 1$, general case follows trivially. It is easy to verify that the optimal list is constructed by first ordering the subsequence $\mathbf{v}_{n-m_n}$ by most likely to least likely, and then concatenating to each such subsequence all the possible $\mathbf{u}_{m_n}$, in an arbitrary order. To reach a given $\mathbf{x}_n = (v^{n-m_n},u^{m_n})$, it is necessary to reach the subsequences $v^{n-m_n}$, and we have:
	\begin{align}
		\mathbb{E}\left[G(\mathbf{X}_n)\right] &= \mathbb{E}\left[\mathbb{E}\left[G(\mathbf{X}_n)| \mathbf{V}_{n-m_n} \right]\right] \nonumber\\*
		& \hspace{-1.5em} \de \sum_{\mathbf{v}_{n-m_n}} \hspace{-.5em}\exp\left\{ -(n-m_n) \left[D(\hat{P}_{\mathbf{v}} || P_V) + H(\hat{P}_{\mathbf{v}})\right] \right\} \nonumber \\* 
		& \hspace{1em} \times \exp\left\{(n-m_n) H(\hat{P}_{\mathbf{v}})\right\} \exp\{m_n\} \nonumber\\
		& \hspace{-1.5em}\de \sum_{\hat{P}_V} \exp\left\{ (n-m_n) \left[ H(\hat{P}_V) - D(\hat{P}_V || P_V) \right] + m_n \right\} \nonumber\\
		& \hspace{-1.5em}\de \exp\left\{ n \sup_{\hat{P}_V} (1 - \lambda)\left[ H(\hat{P}_V) - D(\hat{P}_V||P_V) \right] + \lambda \right\} . \nonumber
	\end{align}
	Solving the optimization yields the desired result.
\end{proof}
The next lemma compares the guesswork of a random variable which takes values in a discrete alphabet uniformly at random, with a random variables for which one of the symbol has been \emph{softly} removed. Precisely, we have
\begin{lemma}[Soft Elimination]\label{lem:soft_elim}
	Consider a random variable $U_N$ taking values uniformly in $[N]$, and $U$. For some $0 \leq s < 1$, we call a $K$ soft-elimination, a random variable $V_{(N,K)}$ such that:
	\begin{align}
	Pr(V_{(N,K)} = i) = \left\{ \begin{array}{ll}
	\frac{1}{N-1} & \text{if } 1 \leq i \leq N-K \\
	\frac{K-1}{K(N-1)} & \text{if } N-K \leq i \leq N 
	\end{array} \right.  .
	\end{align}
	Then, for any $\alpha > 0$, $\mathbb{E}[G(U_{N})^\alpha] > \mathbb{E}[G(V_{(N,K)})^\alpha] \geq \mathbb{E}[G(U_{N-1})^\alpha]$.
\end{lemma}

\begin{proof}
	We have :
	\begin{align}
	& \mathbb{E}[G(V_{(N,K)})] - \mathbb{E}[G(U_{N-1})] =\\*
	 &\hspace{1em}\sum_{i = 1}^{N-K} i^\alpha \left( \frac{1}{N-1} - \frac{1}{N-1} \right) + \nonumber \\
	& \hspace{.35em} \sum_{i = N-K + 1}^{N-1} i^\alpha  \left( \frac{K-1}{K(N-1)}  - \frac{1}{N-1} \right) + N^\alpha \frac{K-1}{K(N-K)}. \nonumber 
	\end{align}
	By evaluating the series and combining terms it is easy to verify that the right hand side is non-negative.
\end{proof}

The following two lemmas relate the position of a sequence $\mathbf{x}$ in the optimal list, i.e. $G^*(\mathbf{x})$, with the type $\hat{P}_{\mathbf{x}}$ of that sequence, first without side-information, and then with side-information.
\begin{lemma}\label{lem:types}
	Let $\mathbf{x}$ be a i.i.d. generated sequence of length $n$,and consider the position of $\mathbf{x}$ in the optimal list according to $P_{X}$, \emph{i.e.} $G^*(\mathbf{x})$. 
	For a given $\alpha$, we have that $G^*(\mathbf{x}) < \lceil |\mathcal{X}|^\alpha \rceil $ if and only if the sequence $\mathbf{x}$ satisfy $\hat{P}_{\mathbf{x}} \in \mathcal{Q}(\alpha)$, where
	\begin{align}
	\mathcal{Q}(\alpha) &= \left\{Q_{X} : D(Q_{X}\| P_{X}) + H(Q_{X}) \right.\nonumber\\
	&\left.\quad\quad\quad\quad\quad< D(Q^*_{X}\| P_{X}) + H(Q^*_{X}) \right\},
	\end{align}
	with $Q^*_{X}$ being the solution of the optimization problem:
	\begin{equation}\label{eq:threshold_prob2}
	\begin{aligned}
	& \underset{Q_{X}}{\text{minimize}}
	& & D(Q_{X} \| P_{X} ) + H(Q_{X}) \\
	& \text{subject to}
	& & H(Q_{X}) \geq \alpha
	\end{aligned}
	\end{equation}
\end{lemma}

\begin{lemma}\label{lem:types_si}
	Let $(\mathbf{x}_n,\mathbf{y}_n)$ be a pair of binary sequences ,and consider the position of $\mathbf{x}$ in the optimal list according to $P_{X|Y}$, \emph{i.e.} $G^*(\mathbf{x}|\mathbf{y})$. 
	For a given $\alpha$, we have that $G^*(\mathbf{x}|\mathbf{y}) < \lceil |\mathcal{X}|^\alpha \rceil $ if and only if the sequence $(\mathbf{x},\mathbf{y})$ satisfy $\hat{P}_{\mathbf{x}|\mathbf{y}} \in \mathcal{Q}(\alpha,\hat{P}_\mathbf{y})$, where
	\begin{align}
	\mathcal{Q}(\alpha,\hat{P}_{\mathbf{y}}) &= \left\{ Q_{X|Y} : D(Q_{X|Y}\| P_{X|Y}|\hat{P}_{\mathbf{y}}) + H(Q_{X|Y}|\hat{P}_{\mathbf{y}})\right.\nonumber\\
	&\left.< D(Q^*_{X|Y}\| P_{X|Y}|\hat{P}_{\mathbf{y}}) + H(Q^*_{X|Y}|\hat{P}_{\mathbf{y}}) \right\}.
	\end{align}
	with $Q^*_{X|Y}$ being the solution of the optimization problem:
	\begin{equation}\label{eq:threshold_prob3}
		\begin{aligned}
		& \underset{Q_{X|Y}}{\text{minimize}}
		& & D(Q_{X|Y} \| P_{X|Y} | \hat{P}_{\mathbf{y}}) + H(Q_{X|Y}|\hat{P}_{\mathbf{y}}) \\
		& \text{subject to}
		& & H(Q_{X|Y}|\hat{P}_{\mathbf{y}}) \geq \alpha .
		\end{aligned}
	\end{equation}
\end{lemma}
Since Lemma~\ref{lem:types} is a direct consequence of Lemma~\ref{lem:types_si}, we only include the proof of the latter.
\begin{proof}
	Recall that $P_{Y|X}(\mathbf{x}|\mathbf{y}) = \exp \{-n \left(D(\hat{P}_{\mathbf{x}|\mathbf{y}}\| P_{X|Y} | \hat{P}_\mathbf{y}) + H(\hat{P}_{\mathbf{x}|\mathbf{y}}|\hat{P}_\mathbf{y}) \right)\}$. 
	Furthermore, note that for a given $\mathbf{y}$ the number of sequences $\mathbf{x}$ which have conditional type $Q_{\mathbf{x}|\mathbf{y}}$ is given by $|T(Q_{\mathbf{x}|\mathbf{y}}) (\mathbf{y})| \de \exp \{ n H(Q_{\mathbf{x}|\mathbf{y}} | \hat{P}_{\mathbf{y}} )  \}$ (see, e.g., \cite[Lemma~2.5]{CsisKro}). 
	Let $\hat{\mathcal{Q}}(\alpha,\hat{P}_{\mathbf{y}})$ be the set of types of the sequences that are in the first $\mathcal{X}^{n\alpha}$ position in the list, that is $\hat{\mathcal{Q}}(\alpha,\hat{P}_{\mathbf{y}})$ is such that:
	\begin{align}
		\sum_{Q_{X|Y} \in \mathcal{Q}(\alpha, \hat{P}_{\mathbf{y}})} 2^{nH(Q_{X|Y}|\hat{P}_{\mathbf{y}})} = 2^{n\alpha} .
	\end{align}
	An application of the method of types gives that the left-hand side evaluates (exponentially) to $2^{n \sup_{Q_{X|Y} \in \hat{\mathcal{Q}}(\alpha,\hat{P}_{\mathbf{y}})} H(Q_{X|Y}|\hat{P}_{\mathbf{y}})}$, meaning that $\sup_{Q_{X|Y} \in \mathcal{Q}}H(Q_{X|Y}|\hat{P}_{\mathbf{y}}) = \alpha$. Thus, the threshold probability is given by \eqref{eq:threshold_prob3}, and any type that has lower probability appears before in the list.
\end{proof}

The list $\mathcal{Q}(\alpha,\hat{P}_{\mathbf{y}})$ is specified implicitely for any $P_{Y|X}$, but can also be made explicit for some specific channels. In particular, binary erasures channels yield to an easy characterization of $\mathcal{Q}(\alpha,\hat{P}_{\mathbf{y}})$. Indeed, in this case the only reverse channel types which need to be considered are those that are valid outputs of an erasure channel. Thus, the order of $\mathbf{x}$ in the ordered list after observation $\mathbf{y}$ solely depends on the type of $\mathbf{x}$ over the position which are erased in $\mathbf{y}$, which we shall denote by $\hat{Q}^{(\epsilon)}_{X|Y}$. Letting $\epsilon$ be the erasure symbol, $\hat{P}_{\mathbf{y}}(\epsilon)$ is thus the fraction of erasures in the received output $\mathbf{y}$, and assuming $P_X(0) > P_X(1)$, we have $\hat{P}_{X|Y} \in \mathcal{Q}(\alpha,\hat{P}_\mathbf{y})$ iff $\hat{Q}^{(\epsilon)}_{X|Y} < \frac{\alpha}{\hat{P}_\mathbf{y}(\epsilon)}$. 

\bibliographystyle{IEEEtran}
\bibliography{strings}

\end{document}